\documentclass[11pt]{article} 

\usepackage{a4wide}
\usepackage{amsmath, amsthm, amsfonts, amssymb, graphicx, color, hyperref, authblk}
\usepackage{complexity}
\usepackage[indention=0mm,font=sl,labelfont=bf]{caption}

\newlang{\searchPM}{\textsc{Search-PM}}
\newlang{\decisionPM}{\textsc{PM}}
\newlang{\MM}{\textsc{MM}}
\newlang{\weightPM}{\textsc{weight-PM}}
\newlang{\minwtPM}{\textsc{Min-Weight-PM}}

\DeclareMathOperator{\conv}{conv} 
\DeclareMathOperator{\sign}{sgn}
\DeclareMathOperator{\Span}{span}
\DeclareMathOperator{\lcm}{lcm}

\newcommand{\QuasiNC}{\text{\rm quasi-}{\sf NC}}
 
\newtheorem{theorem}{Theorem}[section]
\newtheorem{lemma}[theorem]{Lemma}
\newtheorem{corollary}[theorem]{Corollary}

\newtheorem{claim}{Claim}
\newtheorem*{theorem*}{Theorem}
 
\theoremstyle{definition}

\newtheorem*{definition*}{Definition}

\newcommand{\set}[2]{\{\,#1\mid#2\,\}}
\newcommand{\eq}{2em}
\newcommand{\abs}[1]{\lvert #1 \rvert}

\newcommand{\Z}{\mathbb Z}
\newcommand{\xM}{{\boldsymbol x}^M}
\newcommand{\xMp}{{\boldsymbol x}^{M'}}
\newcommand{\x}{{\boldsymbol x}} 
\newcommand{\y}{{\boldsymbol y}}
\newcommand{\w}{{\boldsymbol w}} 
\newcommand{\X}{{\boldsymbol X}}  
\newcommand{\PM}{{\rm PM}} 
\renewcommand{\R}{\mathbb R}
\newcommand{\wc}{c} 
\newcommand{\symdiff}{\mathop{\scriptstyle \triangle}}

\newcommand{\comment}[1]{} 

\title{Bipartite Perfect Matching is in quasi-NC}

\author[1]{Stephen Fenner}
\author[2]{Rohit Gurjar$^*$}
\author[2]{Thomas Thierauf\thanks{Supported by DFG grant TH 472/4}}
\affil[1]{University of South Carolina}
\affil[2]{Aalen University, Germany}

\date{\today}

\begin{document}

\maketitle

\begin{abstract}
We show that the bipartite perfect matching problem is in $\QuasiNC^2$.
That is, it has uniform circuits of quasi-polynomial size~$n^{O(\log n)}$, and $O(\log^2 n)$ depth.
Previously, only an exponential upper bound was known on the size of such circuits with poly-logarithmic depth. 

We obtain our result by an almost complete derandomization of the famous Isolation Lemma
when applied to yield an efficient randomized parallel algorithm for the bipartite perfect matching problem.
\end{abstract}


\section{Introduction}

The perfect matching problem has been widely studied in complexity theory.
It has been of particular interest in the study of derandomization and parallelization. 
The perfect matching problem, $\decisionPM$, asks whether
a given graph contains a perfect matching. 
\comment{
In other words, say for a group of people we are given 
connections of the form `person $x$ and person $y$ know each other' and
the question is whether one can make pairs of persons known to each other such that
every person is paired with someone. 
}

The problem has a polynomial-time algorithm due to Edmonds~\cite{Edm65}.
However, its parallel complexity is still not completely resolved as of today.
The problem can be solved by randomized efficient parallel algorithms due to Lov{\'a}sz~\cite{Lov79}, 
i.e., it is in~$\RNC$,
but it is not known whether randomness is necessary,
i.e., whether it is in~$\NC$.
The class~$\NC$ represents the problems which have efficient parallel algorithms, i.e.,
they have uniform circuits of polynomial size and poly-logarithmic depth.
For the perfect matching problem, nothing better than an exponential-size circuit was known,
in the case of poly-logarithmic depth. 

The construction version of the problem, $\searchPM$, asks to construct a perfect matching
in a graph if one exists.
It is in~$\RNC$ due to 
Karp et al.~\cite{KUW86} and Mulmuley et al.~\cite{MVV87}.
The latter algorithm applies the celebrated Isolation Lemma.
Both algorithms work with a weight assignment on the edges of the graph.
A weight assignment is called \emph{isolating} for a graph~$G$ if 
the minimum weight perfect matching in~$G$  is unique, if one exists.
Mulmuley et al.~\cite{MVV87} showed that 
given an isolating weight assignment with polynomially bounded integer weights for a graph~$G$,  
then a perfect matching in~$G$ can be constructed in~$\NC$.
To get an isolating weight assignment they use randomization.
This is where the Isolation Lemma comes into play.

\begin{lemma}[Isolation Lemma~\cite{MVV87}]
For a graph $G(V,E)$, let $w$ be a random weight assignment, where edges
are assigned weights chosen uniformly and independently at random from $\{1,2, \dots, 2\abs{E}\}$. 
Then~$w$ is isolating with probability~$\geq 1/2$.
\end{lemma}

\emph{Derandomizing} this lemma means to construct such a weight assignment deterministically in~$\NC$.
This remains a challenging open question. 
A general version of this lemma, which considers a family of sets and
requires a unique minimum weight set, has also been studied. 
The general version is related to the polynomial identity testing problem and 
circuit lower bounds~\cite{AM08}.

The Isolation Lemma has been derandomized for some special classes of graphs, 
e.g., planar bipartite graphs~\cite{DKR10,TV12},
 strongly chordal graphs~\cite{DK98}, 
graphs with a small number of perfect matchings~\cite{GK87,AHT07}. 
In this work, we make a significant step towards the derandomization of the Isolation
Lemma for bipartite graphs.
In Section~\ref{sec:isolation}, we construct an isolating weight assignment for these graphs
with quasi-polynomially large weights. 
Previously, the only known deterministic construction was the trivial one that used exponentially large weights.
As a consequence we get that for bipartite graphs, $\decisionPM$ and $\searchPM$ are in $\QuasiNC^2$.
In particular,
they can be solved by uniform Boolean circuits of depth~$O(\log^2 n)$
and size~$n^{O(\log n)}$ for graphs with~$n$ nodes.
Note that the size is just one $\log n$-exponent away from polynomial size.

Our result also gives an $\RNC$-algorithm for $\decisionPM$ in bipartite graphs
which uses very few random bits.
The original $\RNC$-algorithm of Lov{\'a}sz~\cite{Lov79} uses $O(m \log n)$ random bits.
This has been improved by Chari, Rohatgi, and Srinivasan~\cite{CRS95} to $O(n \log (m/n))$ random bits.
They actually construct an isolating weight assignment using these many random bits. 
To the best of our knowledge,
the best upper bound today on the number of random bits is $(n+ n \log (m/n))$ by Chen and Kao~\cite{CK97},
that is, the improvement to~\cite{CRS95} was only in the multiplicative factor.
In Section~\ref{sec:RNC}, we achieve
an \emph{exponential} step down to~$O(\log^2 n)$ random bits.
Note that this is close to a complete derandomization which would be achieved
when the number of random bits comes down to~$O(\log n)$.
This improves an earlier version of this work, 
where we had an $\RNC$-algorithm with~$O(\log^3 n)$ random bits. 

Based on the first version of our paper, Goldwasser and Grossman~\cite{GG15} observed that 
one can get an $\RNC$-algorithm for $\searchPM$ which uses $O(\log^4 n)$ random bits. 
With our improved decision algorithm,
we obtain now an $\RNC$-algorithm for $\searchPM$ which uses only~$O(\log^2 n)$ random bits.

In Section~\ref{sec:extensions} we show that
our approach also gives an alternate $\NC$-algorithm for $\searchPM$ in bipartite planar graphs.
This case already has known $\NC$-algorithms~\cite{MN95,MV00,DKR10}.
Our algorithm is in~$\NC^3$, 
while the previous best known upper bound is already~$\NC^2$~\cite{MN95,DKR10}.

We give a short outline of the main ideas of our approach.
For any two perfect matchings of a graph~$G$, the edges where they differ form disjoint cycles. 
For a cycle~$C$, its circulation is defined to be the difference of weights of two perfect matchings
which differ exactly on the edges of~$C$.
Datta et al.~\cite{DKR10} showed that a weight assignment which ensures nonzero circulation 
for every cycle is isolating. 
It is not clear if there exists such a weight assignment with small weights.
Instead, we use a weight function that has nonzero circulations only for \emph{small} cycles. 
Then, we consider the subgraph~$G'$ of~$G$ which is the union of minimum weight perfect matchings in~$G$. 
In the bipartite case, 
graph~$G'$  is significantly smaller than the original graph~$G$.
In particular, we show that~$G'$ does not contain any cycle with a nonzero circulation.
This means that~$G'$ does not contain any small cycles.

Next, we show that for a graph which has no cycles of length~$< r$,
the number of cycles of length~$<2r$ is polynomially bounded.
This motivates the following strategy which works in~$\log n$ rounds:
in the $i$-th round, assign weights which ensure nonzero circulations for all cycles with length $<2^i$.
Since the graph obtained after $(i-1)$-th rounds has no cycles of length~$<2^{i-1}$,
the number of cycles of length~$<2^i$ is small.
In~$\log n$ rounds, we get a unique minimum weight perfect matching.

\section{Preliminaries}

\subsection{Matchings and Complexity}
By $G(V,E)$ we denote a graph with vertex set~$V$ of size $|V| =n$
and  edge set~$E$ of size $|E| =m$.
We consider only undirected graphs in this paper.
A graph is \emph{bipartite} if 
there exists a partition $V = L \cup R$ of the vertices such that
all edges are between vertices of~$L$ and~$R$.

In a graph $G(V,E)$, a \emph{matching} $M \subseteq E$ is a subset of edges with
no two edges sharing an endpoint.
A matching which covers every vertex is called a \emph{perfect matching}.
For any weight assignment $w \colon E \to \Z$ on the edges of a graph, the
\emph{weight of a matching}~$M$ is defined to be the sum of weights of all the edges in~$M$,
i.e., $w(M) = \sum_{e \in M} w(e)$.

A weight function~$w$ is called \emph{isolating for}~$G$,
if there is a unique perfect matching of minimum weight in~$G$.

A graph~$G$ is \emph{matching-covered} if each edge in~$G$ participates in some perfect matching.
In the literature, matching-covered is also called 1-\emph{extendable} 
and these notions require~$G$ to be \emph{connected}.
{\bf Note}:
in this paper,
we use matching-covered also for non-connected graphs!

The \emph{perfect matching problem} $\decisionPM$ is to decide
whether a given graph has a perfect matching.
Its construction version $\searchPM$ is to compute a perfect matching of a given graph,
or to determine that no perfect matching exists.
A \emph{bipartite} graph~$G(V,E)$ with vertex partition $V = L \cup R$
can have a perfect matching only when $|L| = |R| = n/2$.
Hence,
when we consider bipartite graphs, we will always assume such a partition.

Analogous to $\NC^k$, Barrington~\cite{Bar92} defined the class $\QuasiNC^k$
as the class of problems which have 
uniform circuits of quasi-polynomial size~$2^{\log^{O(1)} n}$ and poly-logarithmic 
depth~$O(\log^k n)$.
Here, \emph{uniformity} means that local queries about the circuit
can be answered in poly-logarithmic time (see~\cite{Bar92} for details). 
The class $\QuasiNC$ is the union of classes~$\QuasiNC^k$, over all~$k \geq 0$.


\subsection{An $\RNC$ algorithm for $\searchPM$}
\label{sec:mvv}

Let us first recall the $\RNC$ algorithm of Mulmuley, Vazirani \& Vazirani~\cite{MVV87} 
for the construction of a perfect matching ($\searchPM$).
Though the algorithm works for any graph, we will only consider bipartite graphs here.

Let $G$ be a bipartite graph with vertex partitions 
$L =\{u_1, u_2, \dots, u_{n/2}\}$ and $R =\{v_1, v_2, \dots, v_{n/2} \}$,
and weight function~$w$.
Consider the following $n/2 \times n/2$ matrix~$A$ associated with~$G$,
$$A(i,j) = \begin{cases}
	2^{w(e)}, & \text{if } e = (u_i, v_j) \in E, \\
	0, & \text{otherwise.}
\end{cases}$$

The algorithm in~\cite{MVV87} computes the determinant of~$A$.
An easy argument shows that this determinant is
the signed sum over all perfect matchings in~$G$:
\begin{eqnarray}
\det(A) &=& \sum_{\pi \in S_{n/2}} \sign(\pi) \prod_{i=1}^{n/2} A(i,\pi(i))
\label{eq:det}\\
&=& \sum_{M \text{ pm in } G} \sign(M)\, 2^{w(M)}
\label{eq:mat}
\end{eqnarray}

Equation~(\ref{eq:mat}) holds because the product $\prod_{i=1}^{n/2} A(i,\pi(i))$
is nonzero if and only if the permuation~$\pi$ corresponds to a perfect matching.
Here $\sign(M)$ is the sign of the corresponding permutation.
If the graph~$G$ does not have a perfect matching, then clearly $\det(A)=0$.
However, even when the graph has perfect matchings, there can be cancellations
due to~$\sign(M)$, and~$\det(A)$ may become zero.
To avoid such cancellations, one needs to design the weight function~$w$ cleverly.
In particular, if~$G$ has a perfect matching
and~$w$ is isolating,
then $\det(A) \neq 0$.
This is because the term~$2^{w(M)}$ corresponding to the minimum weight perfect matching
cannot be canceled with other terms, which are strictly higher powers of~$2$.

Given an isolating weight assignment for~$G$, one can easily construct the minimum 
weight perfect matching in~$\NC$.
Let~$M^*$ be the unique minimum weight perfect matching in~$G$.
First we find out~$w(M^*)$ by looking at the highest power of~$2$ dividing~$\det(A)$.
Then for every edge $e \in E$,
compute the determinant of the matrix~$A_e$ associated with~$G-e$.
If the highest power of~$2$ that divides~$\det(A_e)$ is larger than~$2^{w(M^*)}$, then $e \in M^*$.
Doing this in parallel for each edge, we can find all the edges in~$M^*$.

As already explained in the introduction,
the Isolation Lemma delivers the isolating weight assignment with high probability.
Moreover,
the weights chosen by the Isolation Lemma are polynomially bounded.
Therefore,
the entries in matrix~$A$  have polynomially many bits.
This suffices to compute the determinant in $\NC^2$~\cite{Ber84}.
Hence, also the construction is in~$\NC^2$.
Put together, this yields an $\RNC$-algorithm for $\searchPM$.


\subsection{The Matching Polytope} 
Matchings are also one of the well-studied objects in polyhedral combinatorics. 
Matchings have an associated polytope, called the \emph{perfect matching polytope}.
We use some properties of this polytope to construct an isolating weight assignment.
The perfect matching polytope also
forms the basis of one of the $\NC$-algorithms
for bipartite planar matching~\cite{MV00}.

The perfect matching polytope~$\PM(G)$ of a graph~$G(V,E)$ with $\abs{E} = m$ edges
\index{perfect matching polytope}
is a polytope in the edge space, i.e., $\PM(G) \subseteq \R^{m}$.
For any perfect matching~$M$ of~$G$, 
consider its incidence vector $\xM = (x^M_e)_e \in \R^{m}$ given by
$$x^M_e = \begin{cases} 
		1, & \text{if } e \in M,\\
		0, & \text{otherwise.}
\end{cases}
$$
This vector is referred as a \emph{perfect matching point} for any perfect matching~$M$.
The \emph{perfect matching polytope} of a graph $G$ is defined to be the convex hull 
of all its perfect matching points,
\[
 \PM(G) = \conv\set{\xM}{M \text{ is a perfect matching in } G}.
 \]

Any weight function $w \colon E \to \R$ on the edges of a graph~$G$ can be naturally
extended to~$\R^{m}$ as follows: for any $\x = (x_e)_e \in \R^m$, define
$$w(\x) = \sum_{e\in E} w(e)\, x_e. $$
Clearly, for any matching $M$, we have $w(M) = w(\xM)$.
In particular,
let~$M^*$ be a perfect matching in~$G$ of minimum weight.
Then
\[
w(M^*) = \min \set{w(\x)}{\x \in \PM(G)}.
\]
The following lemma gives a  simple description of the perfect matching polytope
of a bipartite graph~$G$
which is well known, see for example~\cite{LP86}.

\begin{lemma}
\label{lem:polytope}
Let $G$ be a bipartite graph and $\x = (x_e)_e \in \R^m$.
Then $\x \in \PM(G)$ if and only if 
\begin{eqnarray}
\sum_{e \in \delta(v)} x_e &=& 1 \hspace{\eq} v \in V, \label{eq:PMsum}\\
x_e &\geq& 0  \hspace{\eq} e \in E, \label{eq:PMpositive}
\end{eqnarray}
where $\delta(v)$ denotes the set of edges incident on the vertex $v$.
\end{lemma}
It is easy to see that any perfect matching point will satisfy these two conditions.
In fact, all perfect matching points are vertices of this polytope.
The non-trivial part is to show that any point satisfying these two conditions is in 
the perfect matching polytope~\cite[Chapter~7]{LP86}.
For general graphs, the polytope described by~(\ref{eq:PMsum}) and~(\ref{eq:PMpositive})
can have vertices which are not perfect matchings. 
Thus, the description does not capture the perfect matching polytope for general graphs.


\subsection{Nice Cycles and Circulation}

Let $G(V,E)$ be a graph with a perfect matching.
A cycle~$C $ in~$G$ is a \emph{nice cycle},
if the subgraph $G - C$ still has a perfect matching. 
In other words, 
a nice cycle can be obtained from the symmetric difference of two perfect matchings.
Note that a nice cycle is always an even cycle. 

For a weight assignment~$w$ on the edges, 
the \emph{circulation}~$\wc_w(C)$ of an even length cycle $C = (v_1,v_2, \dots, v_k)$ 
is defined as the alternating sum of the edge weights of~$C$,
\[
\wc_w(C) =  \abs{w(v_1,v_2) - w(v_2,v_3) +w(v_3,v_4)- ~ \cdots~ - w(v_k,v_1)}.
\]
The definition is independent of the edge we start with
because we take the absolute value of the alternating sum.

The circulation  of nice cycles was one crucial ingredient of 
the isolation in bipartite planar graphs given by Datta et al.~\cite{DKR10}.

\begin{lemma}[\cite{DKR10}]
\label{lem:circulation}
Let $G$ be a graph with a perfect matching, and let~$w$ be a weight function such that
all nice cycles in~$G$ have nonzero circulation.
Then the minimum perfect matching is unique.
That is,~$w$ is isolating.
\end{lemma}

\begin{proof}
Assume that there two perfect matchings~$M_1, M_2$ of minimum weight in~$G$.
Their symmetric difference $M_1 \symdiff M_2$ consists of nice cycles.
Let~$C$ be a nice cycle in $M_1 \symdiff M_2$.
By the assumption of the lemma,
we have $\wc_w(C) \not= 0$.
Hence, one can decrease the weight of either~$M_1$ or~$M_2$ by altering it on~$C$.
As~$M_1$ and~$M_2$ are minimal, we get a contradiction.
\end{proof}

We will construct an isolating weight function for bipartite graphs.
However,
our weight function will not necessarily have nonzero circulation on all nice cycles.
We start out with a weight assignment which ensures nonzero
circulations for a small set of cycles in a black-box way, 
i.e., without being able to compute the set efficiently.
The following lemma describes a standard trick for this.

\begin{lemma}[\cite{CRS95}]
\label{lem:smallCycles}
Let $G$ be a graph with~$n$ nodes. 
Then, for any number~$s$, one can construct a set of~$O(n^2 s)$ weight assignments 
with weights bounded by~$O(n^2 s)$,
such that for any set of~$s$ cycles, one of the weight assignments gives nonzero circulation
to each of the~$s$ cycles. 
\end{lemma}

\begin{proof}
Let us first assign exponentially large weights.
Let $e_1, e_2, \dots, e_m$ be some enumeration of the edges of~$G$.
Define a weight function~$w$ by $w(e_i) = 2^{i-1}$, for $i = 1,2, \dots, m$.
Then clearly every cycle has a nonzero circulation. 
However, we want to achieve this with small weights.

We consider the weight assignment modulo small numbers,
i.e., the weight functions $\set{w \bmod j}{2 \leq j \leq t}$
for some appropriately chosen~$t$.
We want to show that for any fixed set of~$s$ cycles $\{C_1,C_2, \dots, C_s\}$, 
one of these assignments will work, when~$t$ is chosen large enough. 
That is, we want
$$\exists j \leq t ~~~\forall i\le s:~ c_{w \bmod j}(C_i) \neq 0. 
$$
This will be true provided
$$\exists j \leq t:~ \prod_{i=1}^s c_w(C_i) \not\equiv 0 \pmod{j}  .$$
In other words,
$$\lcm (2,3,\dots,t) \nmid \prod_{i=1}^s c_w(C_i) .$$
This can be achieved by setting $\lcm (2,3,\dots,t) > \prod_{i=1}^s c_w(C_i) $.
The product $\prod_{i=1}^s c_w(C_i)$ is upper bounded by $2^{n^2 s}$.
Furthermore,
we have
 $\lcm (2,3,\dots,t) > 2^{t}$ for $t\geq 7$ (see~\cite{Nai82}).
Thus, choosing $t =  n^2 s$ suffices. 
Clearly, the weights are bounded by $t = n^2 s$.
\end{proof}


\section{Isolation in Bipartite Graphs}
\label{sec:isolation}

In this section we present our main result, 
an almost efficient parallel algorithm for the perfect matching problem.

\begin{theorem}\label{thm:pm-quasiNC}
For bipartite graphs, $\decisionPM$ and $\searchPM$ are in~$\QuasiNC^2$.
\end{theorem}

Let~$G(V,E)$ be the given bipartite graph.
In the following discussion,
we will assume that~$G$ has perfect matchings.
Our major challenge is to isolate one of the perfect matchings in~$G$ by an appropriate weight function.
As we will see later, if~$G$ does not have any perfect matchings,
then our algorithm will detect this.

Our starting point is Lemma~\ref{lem:circulation} which requires nonzero circulations for all nice cycles.
Recall that the construction algorithm requires the weights to be polynomially bounded.
As the number of nice cycles can be exponential in the number of nodes, 
even the existence of such a weight assignment is not immediately clear.
Nonetheless, Datta et al.~\cite{DKR10} give a construction of such a weight assignment for bipartite planar graphs. 
For general bipartite graphs, this is still an open question.

Our approach is to work with a weight function which gives nonzero circulation to only small cycles. 
Lemma~\ref{lem:smallCycles} describes a way to find such weights.
The cost of this weight assignment is proportional to the number of small cycles. 
Further, it is a black-box construction in the sense that one does not need to know the set of cycles.
It just gives a set of weight assignments such that at least one of them has the desired property. 


\subsection{The union of Minimum Weight Perfect Matchings}
Let us assign a weight function for bipartite graph~$G$ which gives nonzero circulation to all small cycles.
Consider a new graph~$G_1$ obtained by the union of minimum weight perfect matchings in~$G$.
Our hope is that~$G_1$ is significantly smaller than the original graph~$G$.
Note that it is not clear if one can efficiently construct~$G_1$ from~$G$.
This is because the determinant of the bi-adjacency matrix with weights 
in equation~(\ref{eq:det}) from Section~\ref{sec:mvv} can still be zero. 
As we will see,
we do not need to construct~$G_1$; it is just used in the argument. 
Our final weight assignment will be completely black-box in this sense.

Our next lemma is the main reason why our technique is restricted to bipartite graphs.
It implies that the graph~$G_1$ constructed from the minimum weight perfect matchings in~$G$
contains no other perfect matchings than these.
In Figure~\ref{fig:nonbipartite}, we give an example showing that this does not hold in general graphs.
The fact that $G_1$ has only minimum weight perfect matchings is equivalent to saying that 
every nice cycle in~$G_1$ has zero circulation. 
The following lemma actually proves an even stronger statement: \emph{every} cycle in~$G_1$ has zero circulation.


\begin{figure}[htbp]
\begin{center}
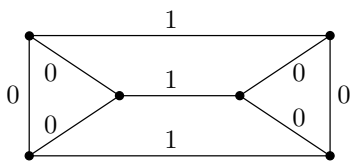
\caption{A non-bipartite weighted graph where every edge is contained in a minimum perfect matching
of weight~$1$.
However, the graph also has a perfect matching of weight~$3$.
That is,
Corollary~\ref{cor:minwtmatchings} does not hold for non-bipartite graphs.
}
\label{fig:nonbipartite}
\end{center}
\end{figure}

\begin{lemma}
\label{lem:cyclesZero}
Let $G(V,E)$ be a bipartite graph with weight function~$w$.
Let~$C$  be a cycle in~$G$ such that $c_w(C) \not= 0$.
Let~$E_1$ be the union of all minimum weight perfect matchings in~$G$.
Then graph $G_1 (V,E_1)$ does not contain cycle~$C$.
\end{lemma}
\begin{proof}
Let the weight of the minimum weight perfect matchings in~$G$ be~$q$.
Let $\x_1,\x_2, \dots, \x_t$ be all the minimum weight perfect matching points of~$G$, i.e., 
the corners of~$\PM(G)$ corresponding to the weight~$q$.
Consider the average point $\x \in \PM(G)$ of these matching points,
$$\x = \frac{\x_1+ \x_2 + \cdots + \x_t}{t}.$$
Clearly, $w(\x) = q$.
Since each edge in~$E_1$ participates in a minimum weight perfect matching, 
for~$\x = (x_e)_e$, we have that $x_e \neq 0$ for all $e \in E_1$.
Now, consider a cycle~$C$ with $c_w(C) \not =0$.
Let the edges of cycle~$C$ be $(e_1,e_2, \dots, e_p)$ in cyclic order. 
For the sake of contradiction let us assume that all the edges of~$C$ lie in~$E_1$.
We show that when we move from point~$\x$ along the cycle~$C$, 
we reach a point in the perfect matching polytope with a weight smaller than~$q$. 
This technique of moving along the cycle has been used by Mahajan and Varadarajan~\cite{MV00}.
To elaborate, consider a new point $\y = (y_e)_e$ such that for all $e \in E$,
$$y_e = \begin{cases}
x_e + (-1)^i\, \varepsilon, & \text{ if } e  = e_i, \text{ for some } 1 \leq i \leq p,  \\
x_e, & \text{ otherwise},
\end{cases}
$$
for some $\varepsilon \neq 0$.
Clearly, the vector $\x -\y$ has nonzero coordinates only on cycle~$C$, 
where its entries are alternating~$\varepsilon$ and~$-\varepsilon$.
Hence, 
\begin{equation}
{w(\x-\y)} = \pm \varepsilon \cdot c_w(C).
\label{eq:XYC}
\end{equation}
As $c_w(C) \neq 0$, we get $w(\x - \y) = w(\x) - w(\y) \neq 0$.
We choose~$\varepsilon  \neq 0$ such that
\begin{itemize}
\item
its sign is such that $w(\y) < w(\x) =q$, and
\item
it is small enough so that $y_e \geq 0$ for all $e \in E$.
This is possible because $x_{e_i} > 0$ for each $1 \leq i \leq p$.
\end{itemize}

We argue that~$\y$ fulfills the conditions of Lemma~\ref{lem:polytope}
and therefore also lies in the perfect matching polytope.
Because $y_e \geq 0$ for all $e \in E$,
it satisfies inequality~(\ref{eq:PMpositive}) from Lemma~\ref{lem:polytope}.
It remains to show that~$\y$ also satisfies 
\begin{equation}\label{eq:y}
\sum_{e \in \delta(v)} y_e = 1 \hspace{\eq} v \in V.
\end{equation}
To see this, let $v \in V$.
We consider two cases:
\begin{enumerate}
\item
$v \not\in C$. 
Then $y_e=x_e$ for each edge $e \in \delta(v)$. 
Thus, we get~(\ref{eq:y}) from equation~(\ref{eq:PMsum}) for~$\x$. 
\item
$v \in C$.
Let~$e_j$ and~$e_{j+1}$ be the two edges from~$C$ which are incident on~$v$.
By definition, 
$y_{e_j} = x_{e_j} + (-1)^j\, \varepsilon$ and 
$y_{e_{j+1}} = x_{e_{j+1}} + (-1)^{j+1}\, \varepsilon$.
For any other edge~$e \in \delta(v)$, we have $y_e = x_e$. 
Combining this with equation~(\ref{eq:PMsum}) for~$\x$, 
we get that~$\y$ satisfies~(\ref{eq:y}) for~$v$.
\end{enumerate}
We conclude that~$\y$ lies in the polytope~$\PM(G)$.
Since $w(\y) < q$,
there must be a corner point of the polytope, which corresponds to a perfect matching in~$G$
with weight~$<q$.
This gives a contradiction.
\end{proof}

After the first version of this paper,
Rao, Shpilka, and Wigderson (see~\cite[Lemma 2.4]{GG15})
came up with an alternate proof of Lemma~\ref{lem:cyclesZero}, 
which is based on Hall's theorem instead of the matching polytope.

A consequence of Lemma~\ref{lem:cyclesZero} is that~$G_1$ has no other perfect matchings
than the ones used to define~$G_1$:
let $M_0,M_1$ be two perfect matchings in~$G_1$.
Their symmetric difference forms a set of cycles. 
By Lemma~\ref{lem:cyclesZero}, 
the circulations of these cycles are all zero.
Hence,
$M_0$ and~$M_1$ have the same weight.

\begin{corollary}
\label{cor:minwtmatchings}
Let $G(V,E)$ be a bipartite graph with weight function~$w$.
Let~$E_1$ be the union of all minimum weight perfect matchings in~$G$.
Then every perfect matching in the graph $G_1 (V,E_1)$ has the same weight -- the minimum weight of
any perfect matching in~$G$.
\end{corollary}
 
Recall that by our weight function, each small cycle in~$G$ has a nonzero circulation.
Therefore by Lemma~\ref{lem:cyclesZero}, $G_1$~has no small cycles. 


Now, we want to repeat this procedure with graph~$G_1$
with a new weight function. 
However, $G_1$ does not have small cycles. 
Hence, we look at slightly larger cycles. 
We argue that their number remains polynomially bounded.

Teo and Koh~\cite{TK92} showed that the number of shortest cycles in a graph with~$m$ edges 
is bounded by~$m^2$.
In the following lemma, 
we extend their argument and give a bound on the number of cycles 
that have length  at most twice the length of shortest cycles.

\begin{lemma}
\label{lem:2gCycles}
Let $H$ be a graph with~$n$ nodes that has no cycles of length~$\leq r$.
Let $r' = 2r$ when~$r$ is even, and $r' = 2r -2$ otherwise.
Then~$H$ has~$\leq n^4$ cycles of length~$\leq r'$.
\end{lemma}

\begin{proof}
Let~$C = (v_0, v_1, \dots, v_{\ell -1})$ be a cycle of length~$\ell \leq r'$ in~$G$.
Let $f = \ell / 4 $.
We successively choose four nodes on~$C$ with distance~$\leq \lceil f \rceil \leq r/2$
and \emph{associate} them with~$C$.
We start with $u_0 = v_0$ and define 
$u_i = v_{\lceil if \rceil}$, for $ i = 1,2,3$.
Note that the distance between~$u_3$ and~$u_0$ is also~$\leq \lceil f \rceil$.
Since we could choose any node of~$C$ as starting point~$u_0$,
the four nodes $(u_0,u_1,u_2,u_3)$ associated with~$C$ are not uniquely defined.
However,
they uniquely describe~$C$.
\begin{claim}\label{cl:unique-cycle}
Cycle~$C$ 
is the only cycle in~$H$ of length~$\leq r'$ that is associated with $(u_0, u_1,u_2,u_3)$.
\end{claim}

\begin{proof}
Suppose~$C' \not= C$ would be another such cycle.
Let~$p \not=p'$ be paths of~$C$ and~$C'$, respectively,
that connect the same $u$-nodes.
Note that~$p$ and~$p'$ create a cycle  in~$H$ of length at most
\[
 |p| + |p'| ~\leq~  {r \over 2} + {r \over 2}  ~\leq~ r, 
 \]
which is a contradiction.
This proves the claim.
\end{proof}

There are~$\leq n^4$ ways to choose~4 nodes and their order.
By Claim~\ref{cl:unique-cycle}, this gives a bound on the number of cycles of length~$\leq r'$.
\end{proof}

Lemma~\ref{lem:2gCycles} suggests the following strategy how to continue from~$G_1$:
in each successive round, we double the length of the cycles
and adapt the weight function to give nonzero circulations to these slightly longer cycles.
By Lemma~\ref{lem:cyclesZero}, 
we have that any cycle with nonzero circulation disappears from the new graph
obtained by taking only the minimum perfect matchings from the previous graph.
Thus, in~$\log n$ rounds we reach a graph with no cycles, i.e., with a unique perfect matching.
Now, we put all the ingredients together and formally define our weight assignment.

\subsection{Constructing the Weight Assignment}
\label{sec:weight}

Let $G(V,E) = G_0$ be bipartite graph with~$n$ nodes that has perfect matchings.
Define $k = \lceil \log n \rceil -1$,  
which is the number of rounds we will need.
We will define subgraphs~$G_i$ and weight assignments~$w_i$,
for $i = 0,1,2, \dots, k-1$,  which will be obtained in successive rounds. 
Note that the shortest cycles have length~$4$. 
Define
\begin{itemize}
\item[$w_i$:] a weight function such that all cycles in~$G_i$ of length~$\leq 2^{i+2}$ 
have nonzero circulations. 
\item[$G_{i+1}$:] the union of minimum weight perfect matchings in~$G_{i}$ according
to weight~$w_{i}$.
\end{itemize}

By the definition of~$G_i$,
any two perfect matchings in~$G_i$ have the same weight,
not only according to~$w_i$, but also to~$w_j$ for all $j < i$, 
for any $1 \leq i \leq k$.

By Lemma~\ref{lem:cyclesZero},
graph~$G_{i}$ does not have any cycles of length $\leq 2^{i+1}$ for each $1 \leq i \leq k$.
In particular, $G_{k}$ does not have any cycles,
since $2^{k+1} \geq n$.
Therefore~$G_k$ has a unique perfect matching.

Our final weight function~$w$ will be a combination of $w_0,w_1, \dots,w_{k-1}$. 
We combine them in a way that 
the weight assignment in a later round does not interfere with the order of perfect matchings
given by earlier round weights. 
Let~$B$ be a number greater than the weight of any edge under any of these weight assignments.
Then, define
\begin{equation}
\label{eq:w}
w = w_0 B^{k-1} + w_1 B^{k-2} + \cdots + w_{k-1} B^0. 
\end{equation}

In the definition of~$w$,
the precedence decreases from~$w_0$ to~$w_{k-1}$.
That is, for any two perfect matchings~$M_1$ and~$M_2$ in~$G_0$,
we have $w(M_1) < w(M_2)$, if and only if there exists an $0 \leq i \leq k-1$ such that
\begin{eqnarray*}
w_j(M_1) &=& w_j(M_2) , \hspace{\eq} \text{for } j < i, \\
 w_i(M_1) &<& w_i(M_2).
\end{eqnarray*}

As a consequence,
the perfect matchings left in~$G_i$ have a 
strictly smaller weight with respect to~$w$ than the ones in~$G_{i-1}$ that did not make it to~$G_i$.

\begin{lemma}
\label{lem:weightM1M2}
For any $1\leq i \leq k$, let~$M_1$ be a perfect matching in~$G_i$
and~$M_2$ be a perfect matching in~$G_{i-1}$ which is not in~$G_i$. 
Then $w(M_1) < w(M_2)$.
\end{lemma}
\begin{proof}
Since~$M_1$ and~$M_2$ are perfect matchings in~$G_{i-1}$,
we have $w_j(M_1) = w_j (M_2)$, for all $j < i-1$, as observed above.
From the definition of~$G_i$ and Corollary~\ref{cor:minwtmatchings}, it follows that 
$w_{i-1}(M_1) < w_{i-1}(M_2).$
 Hence we get that $w(M_1) < w(M_2)$.
\end{proof}

It follows that 
the unique perfect matching in~$G_k$ has a strictly smaller weight with respect to~$w$
than all other perfect matchings.

\begin{corollary}\label{cor:isolatingG_0}
The weight assignment~$w$ defined in~{\rm (\ref{eq:w})} is isolating for~$G_0$.
\end{corollary}

It remains to bound the values of the weights assigned. 
Let us look at the number of cycles which need to be assigned a nonzero circulation in each round.
In the first round, we give nonzero circulation to all  cycles of length~$4$.
Clearly, the number of such cycles is~$ \leq n^4$.
In the $i$-th round, 
we have graph~$G_i$ that does not have any cycles of length~$\leq 2^{i+1}$.
For~$G_i$, we give nonzero circulation to all cycles of length~$\leq 2^{i+2}$.
By Lemma~\ref{lem:2gCycles}, the number of such cycles is~$\leq n^4$.
Therefore, each~$w_i$ needs to give nonzero circulations to~$\leq n^4$ cycles, 
for $0 \leq i < k$.

Now  we apply Lemma~\ref{lem:smallCycles} with $s = n^4$.
This yields
a set of~$O(n^6)$ weight assignments  with weights  bounded by~$O(n^6)$.
Recall that the number~$B$ used in equation~(\ref{eq:w}) is the highest weight assigned by any~$w_i$.
Hence, we also have $B = O(n^{6})$. 
Therefore
the weights in the assignment~$w$ in equation~(\ref{eq:w}) are bounded by~$B^k = O(n^{6 \log n})$.
That is, the weights have~$O(\log^2 n)$ bits. 

For each~$w_i$ we have~$O(n^6)$ possibilities and we do not know which one would work.
Therefore we try all of them.
In total, we need to try $O(n^{6k}) = O(n^{6 \log n})$ weight assignments. 
This can be done in parallel.

Clearly,
every weight assignment can be constructed in~$\QuasiNC^1$
with circuit size~$2^{O(\log^2 n)}$.

\begin{lemma}
In~$\QuasiNC^1$, one can construct a set of~$O(n^{6 \log n})$ integer weight functions
on $[n/2]\times[n/2]$,
where the weights have~$O(\log^2 n)$ bits, such that 
for any given bipartite graph with~$n$ nodes, 
one of the weight functions is isolating.
\end{lemma}

With this construction of  weight functions, we can decide the 
existence of a perfect matching in a bipartite graph in~$\QuasiNC^2$ as follows:
Recall the bi-adjacency matrix~$A$ from Section~\ref{sec:mvv}
which has entry~$2^{w(e)}$ for edge~$e$.
We compute~$\det(A)$ for each of the constructed weight functions
in parallel. 
If the given graph has a perfect matching, then one of the weight functions 
isolates a perfect matching.
As we discussed in Section~\ref{sec:mvv}, for this weight function~$\det(A)$ will be nonzero.
When there is no perfect matching,
then~$\det(A)$ will be zero for any weight function.

As our weights have~$O(\log^2 n)$ bits, the determinant entries have quasi-polynomial bits. 
The determinant can still be computed in parallel, 
with circuits of quasi-polynomial size~$2^{O(\log^2 n)}$
by the algorithm of Berkowitz~\cite{Ber84}.
As we need to compute $2^{O(\log^2 n)}$-many determinants in parallel,
our algorithm is in~$\QuasiNC^2$ with circuit size~$2^{O(\log^2 n)}$.

To construct a perfect matching, we follow
the algorithm of Mulmuley et al.~\cite{MVV87} from Section~\ref{sec:mvv} with each of our weight functions.
For a weight function~$w$ which is isolating, 
the algorithm outputs the unique minimum weight perfect matching~$M$.
If we have a weight function~$w'$ which is not isolating,
still~$\det(A)$ might be non-zero with respect to~$w'$.
In this case,
the algorithm computes a set of edges~$M'$ that
might or might not be a perfect matching. 
However,
it is easy to verify if~$M'$ is indeed a perfect matching,
and in this case, we will output~$M'$.
As the algorithm involves computation of similar determinants as in the decision algorithm,
it is in~$\QuasiNC^2$ with circuit size $2^{O(\log^2 n)}$.
This finishes the proof of Theorem~\ref{thm:pm-quasiNC}.


\section{An $\RNC$-Algorithm with Few Random Bits}
\label{sec:RNC}

We can also present our result for bipartite perfect matching in an alternate way.
Instead of $\QuasiNC$, we can get an $\RNC$-circuit
but with only poly-logarithmically many, namely~$O(\log^2 n)$ random bits.
Note that for a complete derandomization,
it would suffice to bring the number of random bits down to~$O(\log n)$.
Then there are only polynomially many random strings which can all be tested in~$\NC$.
Hence we are only one log-factor away from a complete derandomization.

\subsection{Decision Version}
First, let us look at the decision version.


\begin{theorem}\label{thm:pm-logRNC}
For bipartite graphs, there is an $\RNC^2$-algorithm for $\decisionPM$
which uses~$O(\log^2 n)$ random bits. 
\end{theorem}

To prove Theorem~\ref{thm:pm-logRNC},
consider our algorithm from Section~\ref{sec:isolation}. 
There are two reasons that we need quasi-polynomially large circuits: 
(i) we need to try quasi-polynomially many different weight assignments
and (ii) each weight assignment has quasi-polynomially large weights. 
We show how to come down to polynomial bounds in both cases by using randomization.

To solve the first problem, we modify Lemma~\ref{lem:smallCycles} to
get a random weight assignment which works with high probability.

\begin{lemma}[\cite{CRS95,KS01}]
\label{lem:smallCyclesrandom}
Let~$G$ be a graph with~$n$ nodes and $s \geq 1$. 
There is a random weight assignment~$w$ 
which uses~$O(\log n s)$ random bits and assigns weights  bounded by~$O(n^3 s \log ns)$,
i.e., with~$O(\log ns)$ bits,
such that for any set of~$s$ cycles, $w$ gives nonzero circulation
to each of the~$s$ cycles with probability at least~$1-1/n$. 
\end{lemma}

\begin{proof}
We follow the construction of Lemma~\ref{lem:smallCycles} and
give exponential weights modulo small numbers. 
Here, we use only prime numbers as moduli.
Recall the weight function $w$ defined by $w(e_i) = 2^{i-1}$.
Let us choose a random number $p$ among the first~$t$ prime numbers.
We take our random weight assignment to be $w \bmod p$.
We want to show that with high probability this weight function gives
nonzero circulation to every cycle in $\{C_1,C_2, \dots, C_s\}$.
In other words, $\prod_{i=1}^s c_w(C_i) \not\equiv 0 \pmod p$.
As the product is bounded by~$2^{n^2 s}$, it has at most~$n^2 s$ prime factors. 
Let us choose $t = n^3 s$.
This would mean that a random prime works with probability at least~$(1-1/n)$.
As the $t$-th prime can only be as large as~$2 t \log t$, 
the weights are bounded by $2 t \log t = O(n^3 s \log ns)$,
and hence have~$O( \log ns)$ bits.
A random prime with~$O(\log ns)$ bits can be constructed using~$O(\log ns)$ random bits (see~\cite{KS01}).
\end{proof}

Recall from  Section~\ref{sec:weight} that
for a bipartite graph~$G$ with~$n$ nodes, 
we had $k = \lceil \log n \rceil -1$ rounds and constructed one weight function in each round.
We do the same here,
however,
we use the random scheme from Lemma~\ref{lem:smallCyclesrandom} to
choose each of the weight functions $w_0,w_1, \dots, w_{k-1}$ independently.
The probability that all of them 
provide nonzero circulation on their respective set of cycles
$\geq 1-k/n \geq 1 - \log n /n$ using the union bound. 

Now, instead of combining them to form a single weight assignment, we use 
a different variable for each weight assignment. 
We modify the construction of matrix~$A$ from Section~\ref{sec:mvv}.
Let $L =\{u_1, u_2, \dots, u_{n/2}\}$ and $R =\{v_1, v_2, \dots, v_{n/2}\}$ be the vertex partition of~$G$.
For variables $x_0, x_1, \dots, x_{k-1}$, define an $n/2 \times n/2$ matrix~$A$ by
$$A(i,j) = \begin{cases}
	x_0^{w_0(e)} x_1^{w_1(e)} \cdots x_{k-1}^{w_{k-1}(e)}, & \text{if } e = (u_i, v_j) \in E, \\

	0, &\text{otherwise.}
\end{cases}$$
From arguments similar to those in Section~\ref{sec:mvv}, one can write
$$\det(A) = \sum_{M \text{ perfect matching in } G} \sign(M)\, x_0^{w_0(M)} x_1^{w_1(M)} \cdots x_{k-1}^{w_{k-1}(M)},$$
 where~$\sign(M)$ is the sign of the corresponding permutation.
From the construction of the weight assignments 
it follows that if the graph has a perfect matching then
 the lexicographically minimum term in~$\det(A)$, with respect to 
the exponents of variables $x_0, x_1, \dots, x_{k-1}$ in this precedence order, comes from a unique
perfect matching.
Thus, we get the following lemma.
\begin{lemma}
 $\det(A) \neq 0 \iff$ $G$ has a perfect matching.
\end{lemma}

Recall that each~$w_i$ needs to give nonzero circulations to~$n^{4}$ cycles.
Thus, the weights obtained by the scheme of Lemma~\ref{lem:smallCyclesrandom} will be bounded by~$O(n^7 \log n)$.
This means the weight of a matching will be bounded by~$O(n^8 \log n)$. 
Hence~$\det(A)$ is a polynomial of individual degree~$O(n^8 \log n)$ with~$\log n$ variables.
To test if~$\det(A)$ is nonzero 
one can apply the standard randomized polynomial identity test~\cite{Sch80,Zip79, DL78}.
That is, to plug in random values for variables~$x_{i}$, independently from $\{1,2,\dots,  n^9\}$.
If $\det(A) \neq 0$, then the evaluation is nonzero with high probability.

\paragraph{Number of random bits:}
For a weight assignment~$w_i$, we need~$O(\log n s)$ random bits from Lemma~\ref{lem:smallCyclesrandom}, where $s= n^{4}$.
Thus, the number of random bits required for all~$w_i$'s together is $O(k \log n) = O(\log^2 n)$.
Finally, we need to plug in~$O(\log n)$ random bits for each~$x_{i}$. 
This again requires~$O(\log^2 n)$ random bits. 

\paragraph{Complexity:}
The weight construction involves taking exponential weights modulo small primes by Lemma~\ref{lem:smallCyclesrandom}. 
Primality testing can be done by the brute force algorithm in~$\NC^2$, as the numbers involved have~$O(\log n)$ bits.
Thus, the weight assignments can be constructed in~$\NC^2$. 
Moreover, the determinant with polynomially bounded entries can be computed in~$\NC^2$~\cite{Ber84}.

In summary,
we get an $\RNC^2$-algorithm that uses~$O(\log^2 n)$ random bits as claimed in Theorem~\ref{thm:pm-logRNC}.


\subsection{Search Version}

We get a similar algorithm for $\searchPM$ using also only~$O(\log^2 n)$  random bits.
This improves the $\RNC$-algorithm of Goldwasser and Grossman~\cite{GG15}
based on an earlier version of this paper that uses~$O(\log^4 n)$  random bits.
Their $\RNC$-algorithm has an additional property: it is \emph{pseudo-deterministic},
i.e., it outputs the same perfect matching for almost all choices of random bits.
Our algorithm does not have this property.

\begin{theorem}\label{thm:search-pm-logRNC}
For bipartite graphs, there is an $\RNC^3$-algorithm for $\searchPM$
which uses~$O(\log^2 n)$ random bits. 
\end{theorem}

\comment{
We construct the weight assignments $w_0,w_1, \dots, w_{k-1}$ as described in 
Section~\ref{sec:PMRNC} with $O(\log^2 n)$ random bits. 
As discussed earlier, with probability $\geq (1 - \log n /n)$ the weight assignments
provide nonzero circulation to their respective set of cycles and thus the graph $G_k$
has a unique perfect matching. 
}

Let again~$G(V,E)$ be the given bipartite graph
with vertex partition $L =\{u_1, u_2, \dots, u_{n/2}\}$ and $R =\{v_1, v_2, \dots, v_{n/2}\}$.
We construct the weight assignments $w_0,w_1, \dots, w_{k-1}$ as in 
Lemma~\ref{lem:smallCyclesrandom} in the randomized decision version.
Let~$M^*$ be the unique minimum weight perfect matching in~$G$ with respect to the combined weight function~$w$.
Let $w_r(M^*) = w_r^*$, for $0\leq r < k$.

Recall from Section~\ref{sec:weight} the sequence of subgraphs $G_1, G_2,  \dots, G_k$ of~$G = G_0$,
where~$G_{r+1}$ consists of the minimum perfect matchings of~$G_r$ according to weight~$w_r$.
In order to compute~$M^*$,
we would like to actually construct all the graphs $G_1, G_2,  \dots, G_k$.
However,
it is not clear how to achieve this with~$O(\log^2 n)$ random bits.
Instead,
we will construct a sequence of graphs $H_1,H_2, \dots, H_k$
such that~$H_r$ will be a subgraph of~$G_r$, for each $1 \leq r \leq k$.
Furthermore, each~$H_r$ will contain the matching~$M^*$.
Recall that~$G_k$ consists of the unique perfect matching~$M^*$.
Hence, once we have~$H_k = G_k$, we are done.

Let $H_0 = G$ and $0 \leq r < k$.
We describe the $r$-th round. 
Suppose we have constructed the graph~$H_{r}(V,E_r)$ and
want to compute~$H_{r+1}$. 
An edge will appear in~$H_{r+1}$ only if it participates in a matching~$M$
with $w_r(M) = w_r^*$. 
Thus, we will have that~$H_{r+1}$ is a subgraph of~$G_{r+1}$.
For an edge~$e$, let~$\X_r^{\w(e)}$ denote the product 
\[
\X_r^{\w(e)} = x_r^{w_r(e)} x_{r+1}^{w_{r+1}(e)} \cdots x_{k-1}^{w_{k-1}(e)}\, .
\]
For a matching~$M$, the term~$\X_r^{\w(M)}$ is defined similarly.
Let~$N(e)$ denote the set of edges which are neighbors of an edge~$e$ in~$G_r$,
i.e.\ all edges $e'\not= e$ that share an endpoint with~$e$.
For an edge $e \in E_{r}$, define the $n/2 \times n/2$ matrix~$A_{e}$ as
$$A_{e}(i,j) = \begin{cases}
	\X_r^{\w(e')}, & \text{if } e' = (u_i, v_j) \in E_{r} - N(e), \\
	0, &\text{otherwise.}
\end{cases}$$
Note that the matrix~$A_{e}$ has a zero entry for each neighboring edge of~$e$. 
Thus, its determinant is a sum over all perfect matchings which contain~$e$. 
That is, 
\begin{eqnarray*}
\det(A_{e}) &=& \sum_{\substack{ M \text{ pm in } H_r \\ e \in M} } \sign(M)\, \X_r^{\w(M)}\, .
\end{eqnarray*}
Consider the coefficient~$c_{e}$ of~$x_r^{w_r^*}$ in~$\det(A_{e})$, 
$$c_{e} = \sum_{\substack{ M \text{ pm in } H_r \\ w_r(M) = w_r^*, \, e \in M} } \sign(M)\,  \X_{r+1}^{\w(M)}\, .$$
Define the graph~$H_{r+1}$ to be the union of all the edges~$e$ for which 
the polynomial $c_{e} \neq 0$.
We claim that each edge of~$M^*$ appears in~$H_{r+1}$. 
For any edge $e \in M^*$, 
the polynomial~$c_e$ will contain the term~$\X_{r+1}^{\w(M^*)}$.
As the matching~$M^*$ is isolated in~$H_r$ with respect to the weight vector $(w_{r+1}, \dots, w_{k-1})$,
the polynomial~$c_e$ is nonzero. 

For the construction of~$H_{r+1}$,
we need to test if~$c_{e}$ is nonzero,
for each edge~$e$ in~$H_r$. 
As argued above in the decision part, the degree of~$c_{e}$ is~$O(n^8 \log^2 n)$.
We apply the standard zero-test, i.e.,
we plug in random values for the variables $x_{r+1} , \dots, x_{k-1}$ 
independently from $\{1,2, \dots, n^{11}\}$.
The probability that the evaluation will be nonzero is at least $1- O(\log^2 n/n^3)$.
To compute this evaluation, we plug in values of $x_{r+1}, \dots, x_{k-1}$
in~$\det(A_{e})$ and find the coefficient of~$x_r^{w_r^*}$. 
This can be done in $\NC^2$~\cite[Corollary 4.4]{BCP84}.
For all the edges,
we use the same random values for variables $x_{r+1} , \dots, x_{k-1}$ in each identity test. 
The probability that the test works successfully for each edge is at least $1- O(\log^2 n /n)$ by the union bound.
We continue this for~$k$ rounds to find~$H_k$, which is a perfect matching. 

We need again~$O(\log^2 n )$ random bits
for the weight assignments~$w_0, w_1, \dots, w_{k-1}$ and the values for the~$x_{i}$'s. 
Note that we use the same random bits for~$x_i$ in all~$k$ rounds.
This decreases the success probability, which is now at least $1- O(\log^3 n)/n$ by the union bound.

In $\NC^2$,
we can construct the weight assignments and compute the determinants in each round.
As we have $k = O(\log n)$ rounds,
the overall complexity becomes~$\NC^3$.


\section{Extensions and related problems}
\label{sec:extensions}


\subsection{Bipartite Planar Graphs}
\label{sec:bipartiteplanar}

The $\searchPM$ problem already has some known $\NC$-algorithms
in the case of bipartite planar graphs~\cite{MN95,MV00,DKR10}.
The one by Mahajan and Varadarajan~\cite{MV00}
is in $\NC^3$, while the other two are in $\NC^2$. 
Our approach from the previous section can be modified to give an alternate $\NC^3$-algorithm for this case.

The weights in our scheme in Section~\ref{sec:weight} become quasi-polynomial because
we need to combine the different weight functions from~$\log n$ rounds
using a different scale. 
To solve this problem,
we use the fact that in planar graphs,
one can count the number of perfect matchings of a given weight in~$\NC^2$
by the  Pfaffian orientation technique~\cite{Kas67,Vaz89}.
As a consequence,
we can actually construct the graphs~$G_i$ in each round in~$\NC^2$.
Thereby we avoid having to combine the weight functions from different rounds.

In more detail,
in the $i$-th round, we need to compute the union of minimum weight perfect
matchings in~$G_{i-1}$ according to~$w_{i-1}$.
For each edge~$e$, we decide in parallel if
deleting~$e$ reduces the count of minimum weight perfect matchings. 
If yes, then edge~$e$ should be present in~$G_i$.
As it takes~$\log n$ rounds to reach a single perfect matching,
the algorithm is in~$\NC^3$.

\subsection{Weighted perfect matchings and maximum matchings}

A generalization of the perfect matching problem is the
\emph{weighted perfect matching problem} (\weightPM),
where we are given a weighted graph,
and we want to compute a perfect matching of minimum weight.
There is no $\NC$-reduction known from $\weightPM$ to the perfect matching problem.
However, the isolation technique works for this problem as well, when the weights are small integers.
We put the given weights on a higher scale and put the weights constructed by our scheme
in Section~\ref{sec:isolation} on a lower scale. 
This ensures that a minimum weight perfect matching according to the combined weight function
also has minimum weight according to the given weight assignment. 
Our scheme ensures that there is a unique minimum weight perfect matching. 
One can construct this perfect matching following the algorithm of
 Mulmuley et al.~\cite{MVV87} (Section~\ref{sec:mvv}).

\begin{corollary}
For bipartite graphs, $\weightPM$ with quasi-polynomially bounded integer weights is in~$\QuasiNC^2$.
\end{corollary}

The maximum matching problem asks to find a maximum size matching in a given graph. 
It is well known that the maximum matching problem ($\MM$) is $\NC$-equivalent to the perfect matching problem
(see for example~\cite{GKMT13}). 
The equivalence holds for both decision versions and the construction versions.
The reductions also preserve bipartiteness of the graph.
Thus, we get the following corollary. 
\begin{corollary}
For bipartite graphs, $\MM$ is in~$\QuasiNC^2$.
\end{corollary}


\subsection{Other related problems}

There are many problems related to perfect matching (see for example~\cite[Chapter 14 and 15]{KR98}).
We mention some of them.

In a directed graph, a cycle cover is a set of disjoint cycles which covers every vertex.
The \emph{cycle cover problem} asks to decide if a given directed graph has a cycle cover. 
The weighted version is to find a minimum weight cycle cover, in a weighted directed graph. 
There are simple reductions which show that the cycle cover problem is equivalent to 
the bipartite matching problem (see~\cite[Section 15.3]{KR98}). 

\begin{corollary}
Cycle cover and its weighted version with quasi-polynomially bounded weights
are in~$\QuasiNC^2$.
\end{corollary}

The \emph{tree isomorphism problem} is to decide whether two given trees are isomorphic.
Tree isomorphism is known to be in~$\NC$.
A seemingly harder problem is \emph{subtree isomorphism}:
given two trees~$T_1$ and~$T_2$,
one has to decide whether~$T_1$ is isomorphic to a subtree of~$T_2$.
It has been shown that subtree isomorphism is equivalent to the 
bipartite perfect matching problem via $\NC$-reductions~\cite{KL89} . 

\begin{corollary}
Subtree isomorphism is in~$\QuasiNC$.
\end{corollary}

In the \emph{maximum flow problem} we have given a network, a directed graph,
with capacities on the edges, and two nodes~$s$ and~$t$.
The task is to compute a maximum flow from~$s$ to~$t$ in the network.
In general, 
the maximum flow problem is known to be $\P$-complete.
However,
when the capacities are polynomially bounded integers,
then there is an $\NC$-reduction to the bipartite perfect matching problem~\cite{KUW86}. 
When the capacities are quasi-polynomial, the reduction still works, but in~$\QuasiNC$.

\begin{corollary}
Maximum flow with quasi-polynomially bounded integer capacities is in~$\QuasiNC$.
\end{corollary}

 Given a directed graph~$G$ and a node~$s$ of~$G$,
the \emph{depth-first search tree problem} is to construct a tree within~$G$ with root~$s$
that corresponds to conducting a depth-first search of~$G$ starting from~$s$.
There is an $\NC$-reduction to bipartite $\weightPM$ with polynomially bounded weights~\cite{AAK90,AA87}.

\begin{corollary}
A depth-first search tree can be constructed in~$\QuasiNC$.
\end{corollary}


\section*{Discussion}

The major open question remains whether one can do isolation with \emph{polynomially bounded} weights. 
Our construction requires quasi-polynomial weights because
it takes $\log n$ rounds to reach a unique perfect matching and the graphs 
obtained in the successive rounds cannot be constructed. 
To get polynomially bounded weights one needs to circumvent this. 

For non-bipartite graphs, the isolation question is open even in the planar case. 
For this case, our approach fails in its first step:
Corollary~\ref{cor:minwtmatchings} no longer holds as demonstrated in Figure~\ref{fig:nonbipartite}.
Can one assign weights in a way which ensures that the union of minimum weight perfect
matchings is significantly smaller than
the original graph?

It needs to be investigated if our ideas can lead to isolation in other objects.
For example, isolation of paths in a directed graph, which is related to the~$\NL$ versus~$\UL$ question.
 
\subsection*{Acknowledgements}
We would like to thank Manindra Agrawal and Nitin Saxena for their constant encouragement and very helpful discussions.
We thank Arpita Korwar for discussions on some techniques used in Section~\ref{sec:RNC},
and Jacobo Tor{\'a}n for discussions on the number of shortest cycles.
\bibliographystyle{alpha}
\bibliography{matchingQuasiNC}


\end{document}